\documentclass[letterpaper, 10 pt, conference]{ieeeconf}  

\IEEEoverridecommandlockouts                              
\usepackage{graphics} 
\usepackage{graphicx}
\usepackage{mathptmx} 
\usepackage{times} 
\usepackage{amsmath} 
\usepackage{amssymb}  
\usepackage{cite}
\usepackage{sidecap}
\sidecaptionvpos{figure}{c}
\newtheorem{theorem}{Theorem}

\def \Prob {\mathbb{P}} 
\def \A {\mathbf{A}}    
\def \P {\mathbf{P}}    
\def \U {\mathbf{U}}    
\def \G {\mathbf{G}}    
\def \I {\mathbf{I}} 
\def \H {\mathbf{H}}    
\def \E {\mathbf{E}}    

\usepackage{flushend}

\title{\LARGE \bf
Optimal regulation of protein degradation to \\ schedule cellular
events with precision}

\author{Khem Raj Ghusinga and Abhyudai Singh
\thanks{Khem Raj Ghusinga is with the Department of Electrical and Computer Engineering, University of Delaware,
        Newark, DE 19716, USA
        {\tt\small khem@udel.edu}}%
\thanks{Abhyudai Singh is with Faculty of Electrical and Computer Engineering, Biomedical Engineering, Mathematical Sciences, University of Delaware, Newark, DE, 19716 USA
        {\tt\small absingh@udel.edu}}
\thanks{AS is supported by the National Science Foundation Grant DMS-1312926, University of Delaware Research Foundation (UDRF) and Oak Ridge Associated Universities (ORAU).}}%

\begin{document}
\maketitle
\thispagestyle{empty}
\pagestyle{empty}
\begin{abstract}
An important occurrence in many cellular contexts is the crossing of a prescribed threshold by a regulatory protein. The timing of such events is stochastic as a consequence of the innate randomness in gene expression. A question of interest is to understand how gene expression is regulated to achieve precision in event timing. To address this, we model event timing using the first-passage time framework - a mathematical tool to analyze the time when a stochastic process first crosses a specific threshold. The protein evolution is described via a simple stochastic model of gene expression. Moreover, we consider the feedback regulation of protein degradation to be a possible noise control mechanism employed to achieve the precision. Exact analytical formulas are developed for the distribution and moments of the first-passage time. Using these expressions, we investigate for the optimal feedback strategy such that noise (coefficient of variation squared) in event timing is minimized around a given fixed mean time. Our results show that the minimum noise is achieved when the protein degradation rate is zero for all protein levels. Lastly, the implications of this finding are discussed. 
\end{abstract}

\section{Introduction}
The process of gene expression is central to a cell's function: the genetic information is used to make mRNA molecules which are further translated into proteins, the molecular machines to carry out different tasks. An important step in execution of many cellular functions is when a protein attains an effective level \cite{McA97}.  For example, in gene regulatory pathways, the protein expressed by one gene regulates expression of a downstream gene once its level crosses a certain threshold \cite{Lam96,GAA98,DRO02,FeM98}. Other examples include cell-fate decisions which are made when specific regulatory proteins cross their respective critical levels \cite{FeM98,InS92}, and temporal order of gene activation which involves different thresholds of a given regulatory protein \cite{ZMR04,FGL05,Alo06,AKR07}.

The probabilistic nature of biochemical reactions and the fact that copy number of the constituents is small make expression of a gene a stochastic process \cite{BKC03,RaO05,YXR06,RaV08,KEB05,SiS13}. As a result, a population of cells induced at the same time sees difference in the times at which a critical protein level is attained in individual cells. While the stochasticity in timing does offer an advantage in terms of phenotype variability before a cell commits to an irreversible cell fate \cite{YuN13}, in some cases precision in timing is paramount \cite{AKR07}. Although how cells achieve precision in timing is not understood very well, we expect that the gene expression might be regulated for this purpose. How some commonly found regulation motifs affect the noise in gene expression has been examined in several studies \cite{BeS00,Alo06,ElK06,Alo07,SiH09b,TZS07,SiH09a,CRB12}. 

In this paper, the regulation mechanism we analyze is degradation control of the protein via a feedback mechanism \cite{ElK06,CRB12}. Specifically, we investigate what form of feedback control of protein degradation results in minimum noise in event timing (quantified as the coefficient of variation squared, $CV^2$), given a fixed mean time. The event time is modeled as a first-passage time ($FPT$), i.e., the first time at which a prescribed protein level is achieved. Previously, we have studied a similar question in \cite{GhS14, GFS15} assuming the regulation strategy to be self-regulation of transcription rate. In \cite{GhS14}, a expression of a gene in bursts was considered along with the assumption that protein degradation rate is zero. Further, in \cite{GFS15} we revisited the same question assuming a constant degradation of the protein but simplified the production in bursts as  birth-death process. In another work \cite{GhS15}, we considered both constant protein production in bursts and constant decay and developed the formulas for $FPT$ for the model to predict how changing various model parameters affect the $FPT$ statistics. This paper extends the $FPT$ calculations when the degradation rate is an arbitrary function of protein count. Employing exact solutions and semi-analytical approach, we show that noise in $FPT$ is minimized when the protein does not degrade at all. 

Remainder of the paper is organized as follows. In section II, we formulate the stochastic gene expression model wherein a protein is produced in bursts and its degradation is a function of the protein level. In the next section, $FPT$ for this gene expression model is determined. Section IV deals with determining the moments of $FPT$, particularly the expressions of first two moments which are required to compute the noise ($CV^2$). Section V discusses the optimal regulation strategy that minimizes noise in $FPT$. The results are discussed next in section VI. Some calculations to support the arguments in the main paper are provided in the appendices.

\section{Stochastic Gene Expression Model}
Let $x(t)$ denote the protein level at time $t$. We consider expression of a gene as depicted in Fig. \ref{fig:model}. We assume that the promoter is always active and consider the four fundamental processes, namely, transcription (making of mRNAs from gene), translation (production of protein from a mRNA), mRNA degradation and protein degradation. Also, the feedback regulation of protein degradation is formulated by assuming that if the protein level $x(t)=i$, then the degradation rate of one protein molecule is given by $\gamma_i$ where $\gamma_i$ represents an arbitrary function of $i$.  

\begin{figure}[h]
\centering
\includegraphics[width=0.85\linewidth]{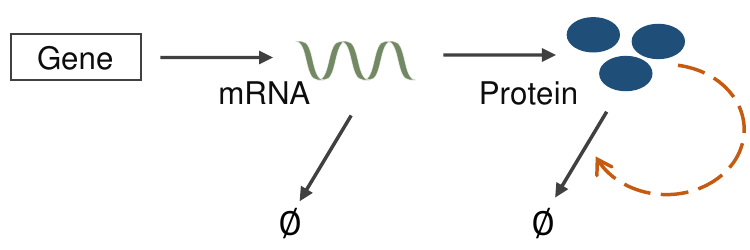}
\caption{\emph{Gene expression model with feedback regulation of protein degradation}. The gene transcribes to make mRNAs which are further translated to proteins. The protein is assumed to regulate its own degradation in order to achieve precision in the time at which a certain protein level is achieved.}
\label{fig:model}
\end{figure}

The mRNA's dynamics can be ignored in this model by assuming that its half-life is considerably smaller than that of the protein - an assumption which is true in most cases \cite{Pau05,SiB12,SRD12,SRC10}. This enables us to reduce the model to what is called a bursty birth-death process wherein each transcription event creates a mRNA molecule which degrades immediately after synthesizing a burst of protein molecules \cite{ShS08,SRD12,SiB12}. The burst size is assumed to follow a geometric distribution which is consistent with previous theoretical and experimental studies \cite{Pau05, FVR05, ShS08, Ber78,Rig79,EJF11,SiB12,SRD12,SRC10}. Mathematically, the probabilities of having an arrival of a burst or a protein degradation in a small time interval $(t,t+dt]$ are given by
\begin{subequations}
\begin{align}
&\mathbb{P}\left( x(t+dt)=i+B | x(t)=i \right) = {k}\;dt, \\
&\mathbb{P}\left( x(t+dt)=i-1 | x(t)=i \right) = i\gamma_i dt.
\end{align}
\end{subequations}
Here $k$ is the transcription rate, $B$ denotes the geometrically distributed burst size while $\mathbb{P}$ represents probability. The probability mass function of $B$ is given by:
\begin{equation}
\label{eqn:goemetricburstpmf}
\mathbb{P}(B=j)=\mu \left(1-\mu \right)^j,\,\mu  \in (0,1),\,j \in \{0,1,2,...\}.
\end{equation}
The mean burst size (average number of protein molecules produced in one mRNA life-time), $b$, can be expressed as:
\begin{equation}
\left<B\right>=b = \frac{1-\mu}{\mu}=\frac{k_p}{{\gamma}_m},
\end{equation}
where $k_p$ and $\gamma_m$ respectively denote the translation rate and the mRNA degradation rate. We are interested in determining the first-time at which the aforementioned model reaches a fix threshold $X$. In the next section, we present these calculations. 
\section{First--passage time calculations}
The first-passage time for the random process $x(t)$ describing protein level to cross a threshold $X$ is defined as:
\begin{equation}
FPT:=\inf\{t:x(t)\geq X\}.
\end{equation}
To compute $FPT$, we employ same approach we used in our previous work \cite{GFS15,GhS15}. This involves considering a particle hopping on an integer lattice. A site on the lattice represents the values protein level can take. The forward jumps correspond to a birth in random bursts (protein production) while a backward jump corresponds to death (protein degradation). The sites with values greater than or equal to $X$ are considered to be absorbing, i.e., the process stops once the particle reaches one of these sites. For such a process, the probability that the threshold $X$ is crossed for the first time in a infinitesimal small time interval $(t,t+dt]$ can be computed as
\begin{align}
\Prob \{FPT \in (t,t+dt)\}
                &=\Prob \{ x(t+dt) \geq X | x(t) \leq X-1\}.
\end{align}
Denoting the probability density function of $FPT$ by $f_{FPT}(t)$, we can write
\begin{subequations}
\begin{align}
f_{FPT}(t)dt
               & =\sum_{i=0}^{X-1}\Prob \left(x(t+dt)\geq X | x(t)=i\right)\mathbb{P}\left(x(t)=i\right), \\
               &= \sum_{i=0}^{X-1}\left(k\Prob\left(B \geq X-i\right)\,dt \right) \Prob \left(x(t)=i\right), \\
               &= \sum_{i=0}^{X-1}k(1-\mu)^{X-i}\Prob\left(x(t)=i\right)\, dt,
\label{eqn:FPTpdfdefI}
\end{align}
\end{subequations}
Here we have made use of $\Prob \left(B \geq X-i \right)=(1-\mu)^{X-i}$ which comes from the fact that distribution of $B$ is geometric with parameter $\mu$ (see equation \eqref{eqn:goemetricburstpmf}). Equation \eqref{eqn:FPTpdfdefI} can also be written as multiplication of two vectors $\U$ and $\P$ as
\begin{equation}
f_{FPT}(t) = \mathbf{U P(t)},
\label{eqn:FPTpdfdefII}
\end{equation}
where the row vector $\U$ and the column vector $\P(t)$ are given as follows
\begin{align}
&\U  ={\begin{bmatrix}
k\left(1-\mu\right)^{X} & k \left(1-\mu\right)^{X-1} & \cdots & k \left(1-\mu\right)
\end{bmatrix}}, \\
&\P (t) = {\begin{bmatrix}
p_0(t) & p_1(t) & \cdots & p_{X-1} (t)
\end{bmatrix}}^T.
\end{align}
In the expression for $\P(t)$ we have used the notation $p_i(t)$ which stands for $\Prob \left(x(t)=i\right)$. To determine $\P(t)$, we write the chemical master equation (or the forward Kolmogorov equation) for the bursty birth-death process with absorbing state\cite{McQ67, GFS15,GhS15}:
\begin{subequations}
\begin{align}
&\dot{p_0}(t) = -k (1-\mu)p_0(t)+\gamma_1  p_1(t),\\
&\dot{p_i}(t) = -\left(k(1-\mu)+i\gamma_i \right)p_i(t)+(i+1)\gamma_{i+1} p_{i+1}(t) \nonumber \\
& \qquad \qquad + \sum_{n=0}^{i-1} k \mu\left(1-\mu\right)^{i-n} p_n(t), \quad 1 \leq i \leq X-2, \\
&\dot{p}_{X-1}(t) = -\left(k (1-\mu)+(X-1)\gamma_{X-1}\right)p_{X-1}(t) \nonumber \\
& \qquad \qquad + \sum_{n=0}^{X-2} k \mu (1-\mu)^{X-n-1} p_n(t).
\end{align}
\end{subequations}
These equations can be expressed as $\dot{\P}(t)=\A \P (t)$ where elements of matrix $\A$ are given by
\begin{equation}
a_{ij} = \begin{cases} 
0, & j>i+1 \\
(i-1)\gamma_{i-1}, & j=i+1 \\
-k\left(1-\mu\right)-(i-1)\gamma_{i-1}, & j=i \\
k \mu \left(1-\mu \right)^{i-j}, & j<i
\end{cases}.
\label{eqn:matrixA}
\end{equation}
The above system of differential equations has the following solution
\begin{equation}
\P(t) = \exp(\A t)\P(0),
\end{equation} 
where $\mathbf{P}(0)={\begin{bmatrix} 1 & 0 & \cdots & 0 \end{bmatrix}}^T$ is the initial probability distribution (as $x(t)=0$ at $t=0$). Using the expression of $\P(t)$, the distribution of FPT can be written as
\begin{equation}
f_{FPT}(t)= \U \exp(\A t)\P(0).
\label{eqn:fptdistribution}
\end{equation}
Equation \eqref{eqn:fptdistribution} gives the probability density function of $FPT$ in terms of known matrices $\U$, $\A$ and $\P(0)$. This expression can be generalized or simplified for a variety of cases such as a feedback regulation of the transcription rate, a different distribution of the burst size, a birth-death process, etc. For each of these cases, the matrix $\U$ and $\A$ will change. Furthermore, $\P(0)$ can also be generalized to other distributions if the initial protein count is not assumed to be zero. In the next section, we use the expression in \eqref{eqn:fptdistribution} to determine expressions of the first two moments of $FPT$.
\section{Moments of FPT} 
\label{sec:FPTmoments}
In this section, we first give the expression of a general $m^{th}$ moment of $FPT$ in terms of the known matrices $\U$, $\A$ and $\P(0)$. Using this result, we further compute the exact formulas for the first two moments of $FPT$. 

The $m^{th}$ moment of $FPT$ can be computed as follows:
\begin{subequations}
\begin{align}
\left<FPT^m\right> &= \int_{0}^{\infty} t^m\, \U \exp(\A t)\P(0) dt, \\
                              &= \U \left(\int_{0}^{\infty} t^m \exp(\A t) dt \right) \P(0). \label{eqn:fptmoments}
\end{align}
\end{subequations}
The matrix $\A$ is a full-rank (invertible) Hurwitz matrix (see Appendices \ref{app:hurwitz}, \ref{app:Ainv}). These properties can be used to get the following expression for a $m^{th}$ order $FPT$ moment \cite{GhS15}
\begin{align}
\left<FPT^m\right> =(-1)^m m! \U (\A^{-1})^{m+1} \P(0), \label{eqn:fptmoment-final}
\end{align}
Using equation \eqref{eqn:fptmoment-final}, we can get formulas for mean and second order moment of $FPT$.
\subsection{Mean $FPT$}
The expression of mean $FPT$ is given by
\begin{align}
\left<FPT\right>&=\U \A^{-2} \P(0).
\end{align}
As discussed in Appendix \ref{app:Ainv}, we have $\A^{-1}=\E^{-1}\A_0^{-1}$ where $\E^{-1}$ and $\A_0^{-1}$ are respectively given by equations \eqref{eqn:Einv} and \eqref{eqn:A0inv}. Therefore, we have
\begin{equation}
\left<FPT\right>=\U \A^{-1}\E^{-1} \A_0^{-1} \P(0)
\end{equation}
The reason for not writing one of the $\A^{-1}$ in terms of $\E^{-1}$ and $\A_0^{-1}$ is that $\U \A^{-1}$ has a much simpler expression as described in Appendix \ref{app:UAinv}.  We further observe that since $\P(0)$ has zero elements excepts for the first one, $\A_0^{-1}\P(0)$ is just the first column of $\A_0$. Therefore $\E^{-1}\A_0^{-1}\P(0)$ is given by
\begin{subequations}
\small{\begin{align}
&-\frac{\mu}{1-\mu}{{\begin{bmatrix}
1            &  \frac{\gamma_1}{k+\gamma_1}\frac{1}{1-\mu}                    & \cdots        & \prod_{l=1}^{j-1}\left(\frac{l\gamma_l}{k+l\gamma_l} \frac{1}{1-\mu}\right)           \\
0           & \frac{k}{k+\gamma_1}                      & \cdots          & \frac{k}{k+\gamma_1} \prod_{l=2}^{j-1}\left(\frac{l\gamma_l}{k+l\gamma_l} \frac{1}{1-\mu}\right) \\
\vdots                          & \vdots                                            & \vdots         & \vdots                                               \\
0           & 0                 & \cdots             & \frac{k}{k+(X-1)\gamma_{X-1}}   \\
\end{bmatrix}}}
{\begin{bmatrix}
\frac{1}{\mu k} \\
\frac{1}{k} \\
\vdots      \\
\frac{1}{k} \\
\end{bmatrix}}, \\
&=-\frac{\mu}{1-\mu}{\begin{bmatrix}
\frac{1-\mu}{\mu}\frac{1}{k}+\frac{1}{k}\left(1+\sum_{j=2}^{X}\prod_{l=1}^{j-1}\left(\frac{l\gamma_l}{k+l\gamma_l} \frac{1}{1-\mu}\right)\right)\\
\frac{1}{k+\gamma_1}\left(1+\sum_{j=3}^{X}\prod_{l=2}^{j-1}\left(\frac{l\gamma_l}{k+l\gamma_l} \frac{1}{1-\mu}\right)\right)      \\
\vdots             \\
\frac{1}{k+(X-1)\gamma_{X-1}}   \\
\end{bmatrix}}.
\end{align}} \normalsize
\end{subequations}
By equation \eqref{eqn:UAinv}, we have $\U \A_0^{-1}=-{\begin{bmatrix} 1 & 1 & \cdots & 1 \end{bmatrix}}$. Therefore the mean $FPT$ is negative of summation of the elements of the vector $\E^{-1}\A_0^{-1}\P(0)$ and is given by
\small{
\begin{equation}
\frac{1}{k}+\frac{\mu}{1-\mu} \sum_{i=1}^{X} \frac{1}{k+(i-1)\gamma_{i-1}} \Bigg(1+ \sum_{j=i+1}^{X} \prod_{l=i}^{j-1}\frac{1}{k+l\gamma_l} \frac{l\gamma_l}{(1-\mu)}\Bigg).
\end{equation}}
\normalsize
\subsection{Second order moment}
The calculations of the second order moment are similar to those of the mean $FPT$. For this reason, we skip the detailed steps and only provide the final formula.
\begin{subequations}
\begin{multline}
\left<FPT^2\right>=2\left(\frac{\mu}{1-\mu}\right)^2\sum_{i=1}^{X}  \left(\frac{1}{k+(i-1)\gamma_{i-1}} \xi_i + \right. \\
\left.\sum_{j=i+1}^{X} \xi_j \left( \frac{1}{k+(i-1)\gamma_{i-1}}\prod_{l=i}^{j-1}\frac{1}{k+l\gamma_l} \frac{l\gamma_l}{(1-\mu)}\right)\right),
\label{eqn:fptsecondmomentdeg}
\end{multline}
\begin{equation}
\text{where} \quad \xi_i:=\frac{1}{k}\left(\sum_{j=1}^{i-1}\eta_j + \frac{\eta_i}{\mu}\right), \quad \text{and} \quad \quad \quad \;\;\;\;\;\;\;\;\;
\label{eqn:xii}
\end{equation}
the terms denoted by $\eta_i$ are given by
\small{
\begin{equation}
\frac{1-\mu}{\mu}\frac{1}{k}\delta_{i-1}+\frac{1}{k+(i-1)\gamma_{i-1}}\left(1 + \sum_{j=i+1}^{X} \prod_{l=i}^{j-1}\frac{1}{k+l\gamma_l} \frac{l\gamma_l}{(1-\mu)}\right).
\end{equation}}
\normalsize
\end{subequations}

Having developed expressions of the moments of $FPT$, we next investigate the optimal feedback regulation of the protein degradation rate that minimizes noise in $FPT$.
\begin{figure}[h]
\centering
\includegraphics[width=\linewidth]{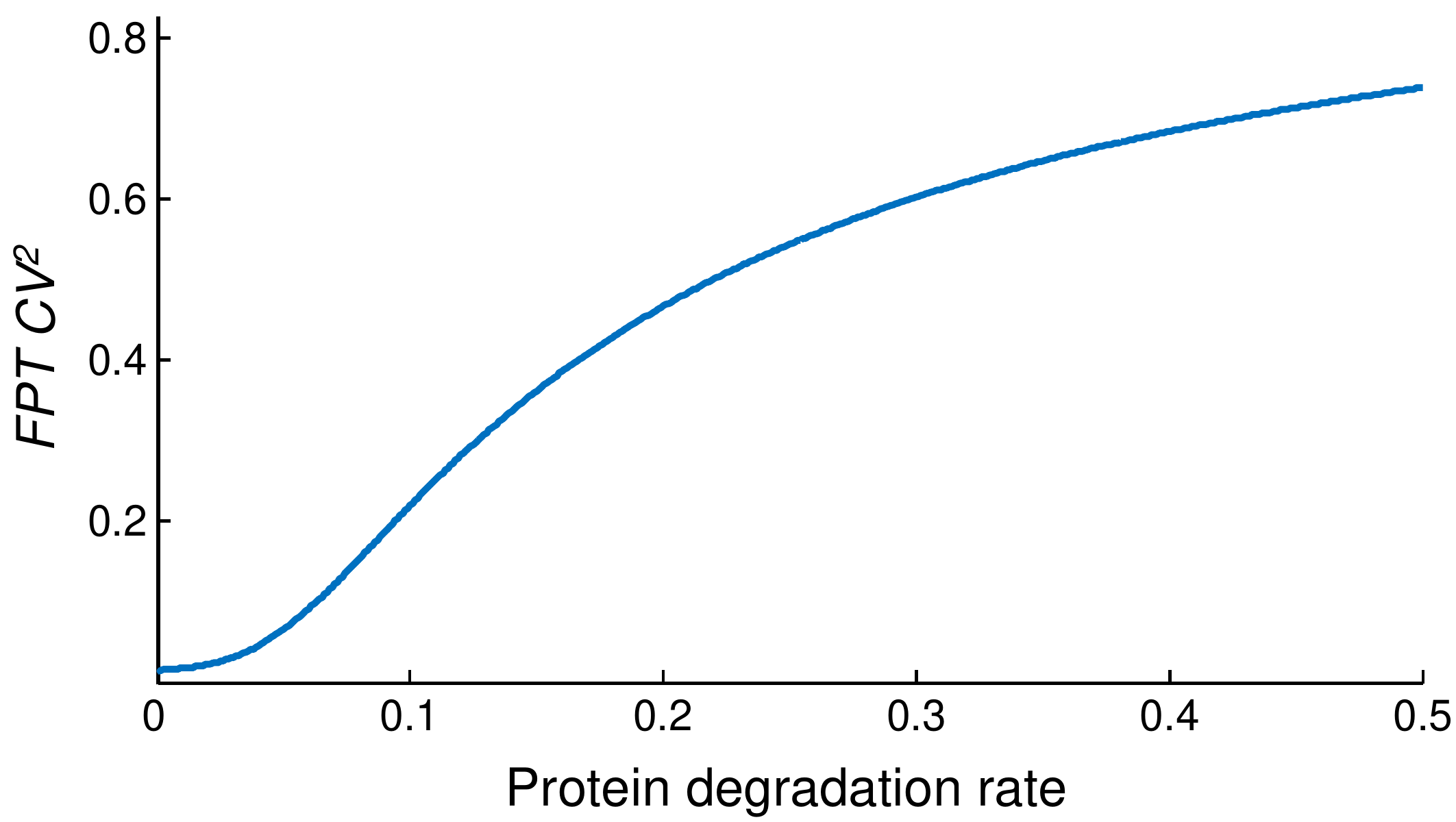}
\caption{\emph{In absence of any feedback regulation, the $CV^2$ of $FPT$ is minimum when the protein degradation rate $\gamma=0$.} The plot shows increase in $CV^2$ as the degradation rate $\gamma$ is increased. The mean $FPT$ is kept constant at 50 minutes by a corresponding adjustment in the transcription rate $k$ for each value of $\gamma$. Values of the parameters used are: threshold $X$=500 molecules, mean burst size $b$=1 molecule.}
\label{fig:openloop}
\end{figure}
\section{Optimal regulation of protein degradation}
We begin with considering an open loop control wherein the protein degradation rate is constant, i.e., $\gamma_i=\gamma$. Using the formulas for the moments of $FPT$ that we developed in the previous section, we plot the $CV^2$ of $FPT$ as the degradation rate is varied while keeping the mean $FPT$ constant with a simultaneous variation in the transcription rate $k$. 

As shown in Fig. \ref{fig:openloop}, the $CV^2$ of $FPT$ increases as $\gamma$ is increased. That is, if there is no feedback regulation of the protein degradation, the best strategy is to have a zero degradation of the protein. It remains to be seen whether or not we can do better by having a feedback control of the protein degradation rate.

We explore this by first analyzing it for a birth-death process which is a simplification of the bursty birth-death process when the mean burst size $b \ll 1$ \cite{GFS15}. 
\subsection{Optimal feedback regulation of protein degradation for a birth-death process}
To analyze the birth-death process, we use the formulas of $FPT$ moments developed in\cite{JoD08}. We would like to point out that in the limit $b \ll 1$ or by carrying out calculation in the same fashion as this paper, we can derive $FPT$ moments (see Appendix \ref{sec:BDPformulas}). The reason of going ahead with the formulas given in \cite{JoD08} is that they are easier to analyze in the present context. The $FPT$ moments are given by
\begin{align}
&\left<FPT\right>=\sum_{i=1}^{X}\frac{1}{k \pi_{i-1}}\sum_{l=0}^{i-1} \pi_{l}, \; \pi_0=1, \; \pi_i=\frac{k^i}{i! \gamma_1 \gamma_2 \cdots \gamma_i}, \\
&var\left(FPT\right)=\sum_{i=1}^{X}\frac{2}{k\pi_{i-1}}\sum_{j=1}^{i-1} \frac{1}{k\pi_{j-1}}\left(\sum_{l=0}^{j-1} \pi_{j-1}\right)^2 \nonumber \\
&\qquad \qquad \qquad \qquad \qquad  \qquad+\frac{1}{k^2 \pi^2_{i-1}}\left(\sum_{l=0}^{i-1} \pi_{l}\right)^2.
\end{align}
We claim that the coefficient of variance squared ($CV^2$, defined as variance over mean squared) is lower bounded for a birth-death process. The result is presented as a theorem.
\begin{theorem}
For a birth-death process with a constant birth rate $k$ and protein level dependent death rate $i\gamma_i$, the $CV^2$ of $FPT$ is lower bounded by ${1}/{X}$ where $X$ is the $FPT$ threshold. The lower bound is achieved only when the degradation rates $\gamma_i$ are zero.
\end{theorem}
\begin{proof}
We first show the equality. Let $\gamma_i=0$ for $i=1,2,\cdots,X-1$. In this case, we have $\left<FPT\right>=X/k$ and $var(FPT)=X/k^2$. Therefore, $CV^2=1/X$. Next, we consider the case when $\gamma_i \neq 0$. Using the Cauchy-Schwarz inequality, we have
\begin{align}
\left<FPT\right>^2=\left(\sum_{i=1}^{X}\frac{1}{k \pi_{i-1}}\sum_{l=0}^{i-1} \pi_{l}\right)^2 \leq X  \sum_{i=1}^{X}\frac{1}{k^2 \pi^2_{i-1}}\left(\sum_{l=0}^{i-1} \pi_{l}\right)^2.
\end{align}
The equality above holds only if
\small{
\begin{equation}
\frac{1}{k^2 \pi^2_{i-1}}\left(\sum_{l=0}^{i-1} \pi_{l}\right)^2=\frac{1}{k^2 \pi^2_{j-1}}\left(\sum_{l=0}^{j-1} \pi_{l}\right)^2, \forall i, j \in \{1,\cdots,X-1\}
\end{equation}}
\normalsize \noindent which is true only in the limit when $\gamma_i \rightarrow 0$. Therefore, when $\gamma_i\neq 0$, we have 
\begin{subequations}
\begin{align}
&\left<FPT\right>^2 < X \sum_{i=1}^{X}\frac{1}{k^2 \pi^2_{i-1}}\left(\sum_{l=0}^{i-1} \pi_{l}\right)^2 \\
\implies &\left<FPT\right>^2 < X \left( \sum_{i=1}^{X}\frac{2}{k \pi_{i-1}}\sum_{j=1}^{i-1} \frac{1}{k\pi_{j-1}}\left(\sum_{l=0}^{j-1} \pi_{j-1}\right)^2\right)\nonumber \\
&\qquad \qquad \qquad\qquad  + X \left(\frac{1}{k^2 \pi^2_{i-1}}\left(\sum_{l=0}^{i-1} \pi_{l}\right)^2 \right) \\
\implies &\left<FPT\right>^2 < X\; var(FPT) \\
\implies & \frac{var(FPT)}{\left<FPT\right>^2}=CV^2 > \frac{1}{X}.
\end{align}
\end{subequations}
This concludes the proof.
\end{proof}

The above result shows that no matter what feedback strategy in degradation is employed, the $CV^2$ in timing will always be greater than $1/X$ which is achieved when the degradation rates for all protein levels is zero. Next, we explore the optimal feedback for the bursty birth-death process.

\subsection{Optimal feedback regulation of protein degradation for a bursty birth-death process}
Intuitively, the result we obtained in the previous section for the birth-death process should not change for the bursty birth-death process. This is because the difference between both these processes is in the birth events and the feedback in death should affect them in similar fashion. We show numerical results to substantiate this argument. We assume that as the protein level is increased, it represses its degradation activity via a negative feedback given by
\begin{equation}
\gamma_i=\frac{\gamma_{\max}}{1+(c\,i)^H}.
\end{equation}
Here $H$ represents the Hill coefficient, $\gamma_{\max}$ is the maximum degradation rate, $i$ is protein level and $c$ is referred to as the feedback strength. In Fig. \ref{fig:negfb}, we show $CV^2$ of $FPT$ as $c$ is varied for given value $\gamma_{\max}$. The mean $FPT$ is kept constant by changing the transcription rate. It can be seen that as the feedback becomes stronger, the $CV^2$ decreases. That is, faster the protein degradation rate decreases, lower the $CV^2$ gets. 
\begin{figure}
\centering
\includegraphics[width=\linewidth]{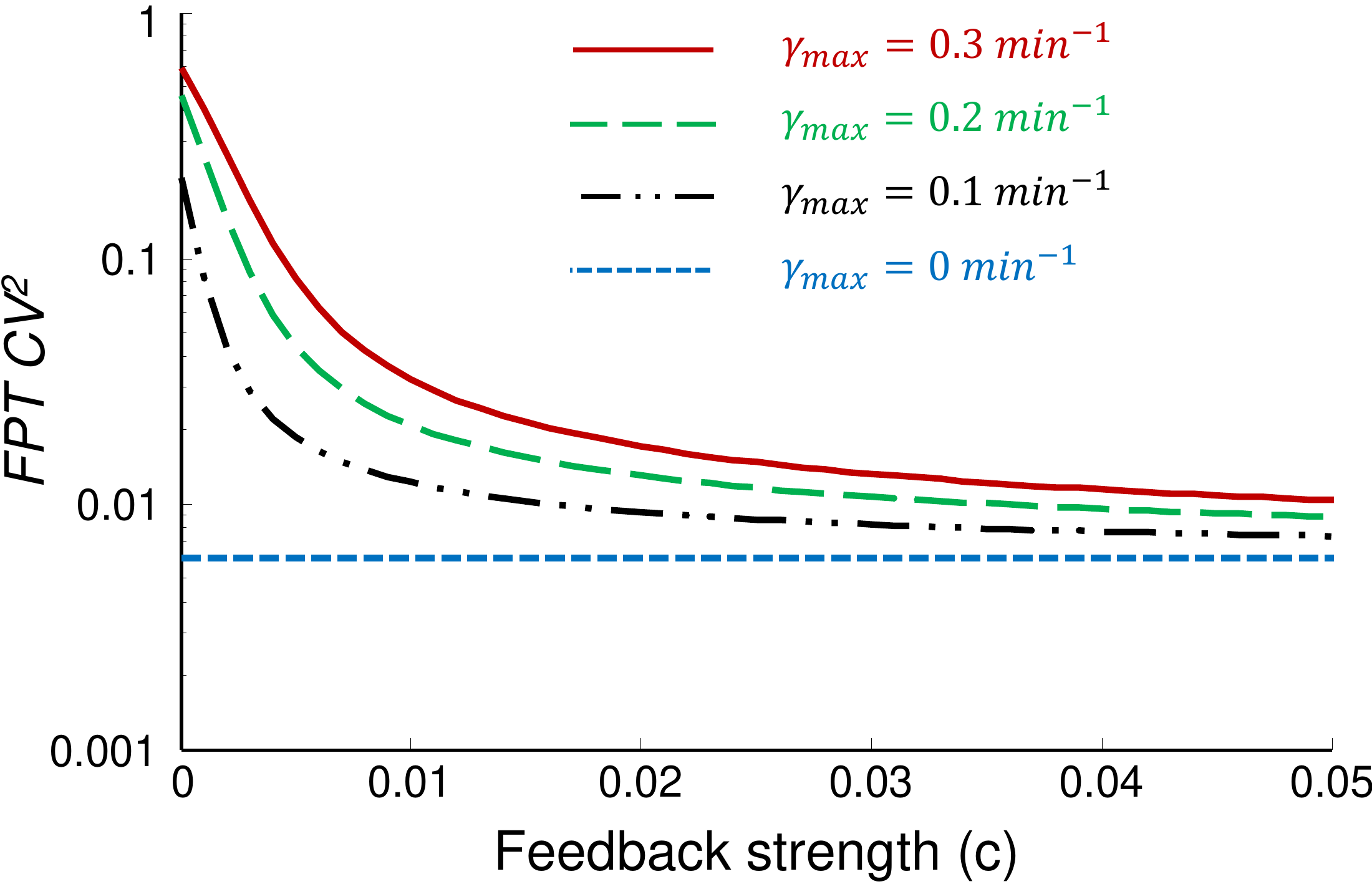}
\caption{\emph{Zero degradation of protein provides minimum noise in timing than the case wherein a negative feedback mechanism is employed when protein degrades.} A negative feedback of the degradation rate is considered as $\gamma_i=\gamma_{max}/(1+c i)$. The plots are shown for different values of $\gamma_{max}$. For each value of the parameter $c$ (feedback strength), a corresponding value of the transcription rate $k$ is computed such that the mean $FPT$ is constant at 50 minutes. Values of the parameters used are: threshold $X$=500 molecules, mean burst size $b$=1 molecule.}
\label{fig:negfb}
\end{figure}

The global minimum in $CV^2$ is the case when the protein does not degrade. For non-zero values of $\gamma_{\max}$, the $CV^2$ approaches this minimum value for very strong negative feedbacks ($c \approx 1$). We have not shown results for the positive feedback to the degradation rate but have checked that the $CV^2$ gets worse as the feedback strength is increased - a result that is expected.

\section{Discussion}
In this work, we characterized the time taken by a protein to reach a certain threshold using the well-known first-passage time framework. Further, we investigated if precision in $FPT$ can be achieved by regulating the degradation rate of the protein while keeping the mean $FPT$ fixed. We showed that this is not possible and the best strategy rather is to have no degradation of the protein. 

This finding complements the results of our previous works in \cite{GhS14,GFS15} wherein we considered the regulation of transcription rate. We showed that when protein does not degrade, the best regulation strategy of the transcription rate is an open loop system \cite{GhS14}. In contrast, if a constant rate of protein decay is considered then the optimal transcription rates seem to be a mixture of both positive and negative feedbacks \cite{GFS15}. Combining these two results, it can be suggested that the best strategy for achieving minimum noise in $FPT$ is to have no degradation of the protein and have open loop (no feedback) production. 

 Questions similar to the one addressed in this work can also be asked about other commonly found gene regulation motifs such as feedforward loops \cite{Alo06}. Further, this work is confined to analyzing a reduced model of gene expression, ignoring the effects of promoter switching \cite{SVK13}, and the parameter regimes when mRNA and protein half-lives are comparable. In future, we would aim to investigate along these lines. 

\appendix
\subsection{Proof that $\A$ is a Hurwitz matrix}
\label{app:hurwitz}
To prove that the matrix defined in equation \eqref{eqn:matrixA} is a Hurwitz matrix, we show that it satisfies the following conditions stated in  \cite{LWY07}:
\begin{enumerate}
\item The diagonal elements $a_{ii} < 0$ for $i=1,2,\cdots,X$,
\item $\displaystyle \max_{1\leq j \leq X} \sum_{\substack{i=1 \\j \neq i}}^{X} \left \lvert \frac{a_{ij}}{a_{jj}} \right \rvert < 1$.
\end{enumerate}
Note that $a_{ii}=-k(1-\mu)-(i-1)\gamma_{i-1} <0$. Therefore, the first requirement is fulfilled. Also, $j=1,2,\cdots,X$
\begin{subequations}
\small{
\begin{align}
\sum_{\substack{i=1 \\j \neq i}}^{X} \left \lvert \frac{a_{ij}}{a_{jj}} \right \rvert &= \frac{(j-1)\gamma_{j-1}}{k(1-\mu)+(j-1)\gamma} \nonumber \\ 
& \qquad \qquad +\; k\,\mu \sum_{i=j+1}^{X} \frac{(1-\mu)^{i-j}}{k(1-\mu)+(j-1)\gamma_{j-1}}\\
&= \frac{k(1-\mu)\left(1-(1-\mu)^{X-j}\right)+(j-1)\gamma_{j-1}}{k(1-\mu)+(j-1)\gamma_{j-1}} < 1.
\end{align}}
\end{subequations}
\normalsize
\noindent That is, $\A$ satisfies the second condition as well. Hence, $\A$ is Hurwitz.

\subsection{Calculations to determine $\A^{-1}$}
\label{app:Ainv}
We use the same procedure as we did in our previous work \cite{GhS15}. We decompose $\A$ as $\A=\A_0+\A_e$ where $\A_0$ consists of only the non-degradation rates $\gamma_i$ terms. Specifically, $\A_0$ and $\A_e$ are respectively given by
\begin{align}
\A_0={\begin{bmatrix}
-k(1-\mu)                 & 0                                          & \cdots                                              & 0              \\
k\mu(1-\mu)            & -k(1-\mu)                      & \cdots        & 0                                \\
k\mu(1-\mu)^2        & k\mu(1-\mu)                               & \cdots                                           & 0               \\
\vdots                          & \vdots                                            & \vdots         & \vdots                                             \\
k\mu(1-\mu)^{X-2}  & k\mu(1-\mu)^{X-3}                     & \cdots         & 0    \\
k\mu(1-\mu)^{X-1}  & k\mu(1-\mu)^{X-2}                     & \cdots                         & -k(1-\mu) \\
\end{bmatrix}},
\end{align}
\begin{align}
\A_e= 
 {\begin{bmatrix}
0            & \gamma_1                                       & \cdots        & 0                                            & 0              \\
0           & -\gamma_1                      & \cdots        & 0                                             & 0              \\
0           & 0                             & \ddots         & 0                                             & 0               \\
\vdots                          & \vdots                                            & \vdots         & \vdots                                      & \vdots          \\
0           & 0                   & \cdots         &  -(X-2)\gamma_{X-2} & (X-1)\gamma_{X-1}     \\
0           & 0                     & \cdots         & 0               & -(X-1)\gamma_{X-1}  \\
\end{bmatrix}}.
\end{align}
Note that inverse of $\A$ can be written as 
\begin{align}
\A^{-1}=\left(\A_0 + \A_e\right)^{-1}= \left(\I + \A_0^{-1}\A_e \right)^{-1}\A_0^{-1}.
\end{align}
The term $\A_0^{-1}$ is easy to determine and is given by
\begin{align}
\A_0^{-1}=-\frac{\mu}{1-\mu}{\begin{bmatrix}
\frac{1}{\mu k}                 & 0                                          & \cdots        & 0                                            & 0              \\
\frac{1}{k}            & \frac{1}{\mu k}                       & \cdots        & 0                                             & 0              \\
\frac{1}{k}           & \frac{1}{k}                              & \cdots         & 0                                             & 0               \\
\vdots                          & \vdots                                            & \vdots         & \vdots                                      & \vdots          \\
\frac{1}{k}   & \frac{1}{k}                     & \cdots         &  \frac{1}{\mu k} & 0    \\
\frac{1}{k}  & \frac{1}{k}                     & \cdots         & \frac{1}{k}                & \frac{1}{\mu k}  \\
\end{bmatrix}}.
\label{eqn:A0inv}
\end{align}
Further, we also need inverse of $\E:=\left(\I+\A_0^{-1}\A_e\right)$ in order to determine $\A^{-1}$. To this end, let us compute $\A_0^{-1}\A_e$: 
\begin{subequations}
\begin{align}
&-\frac{\mu}{1-\mu}\left\{{{\begin{bmatrix}
\frac{1}{\mu k}                 & 0                                          & \cdots        & 0                                            & 0              \\
\frac{1}{k}            & \frac{1}{\mu k}                       & \cdots        & 0                                             & 0              \\
\frac{1}{k}           & \frac{1}{k}                              & \cdots         & 0                                             & 0               \\
\vdots                          & \vdots                                            & \vdots         & \vdots                                      & \vdots          \\
\frac{1}{k}   & \frac{1}{k}                     & \cdots         &  \frac{1}{\mu k} & 0    \\
\frac{1}{k}  & \frac{1}{k}                     & \cdots         & \frac{1}{k}                & \frac{1}{\mu k}  \\
\end{bmatrix}}} \right.\\ 
& \qquad \left. {\begin{bmatrix}
0            & \gamma_1                                        & \cdots        & 0                                            & 0              \\
0           & -\gamma_1                      & \cdots        & 0                                             & 0              \\
0           & 0                             & \ddots         & 0                                             & 0               \\
\vdots                          & \vdots                                            & \vdots         & \vdots                                      & \vdots          \\
0           & 0                   & \cdots         &  -(X-2)\gamma_{X-2} & (X-1)\gamma_{X-1}    \\
0           & 0                     & \cdots         & 0               & -(X-1)\gamma_{X-1} \\
\end{bmatrix}}\right\}\\
&= {\begin{bmatrix}
0            & -\frac{\gamma_1}{(1-\mu) k}                                        & \cdots        & 0                                            & 0              \\
0           & \frac{\gamma_1}{k}                      & \cdots        & 0                                             & 0              \\
0           & 0                             & \ddots         & 0                                             & 0               \\
\vdots                          & \vdots                                            & \vdots         & \vdots                                      & \vdots          \\
0           & 0                   & \cdots         &  \frac{(X-2)\gamma_{X-2}}{k}& -\frac{(X-1)\gamma_{X-1}}{(1-\mu)k}    \\
0           & 0                     & \cdots         & 0               & \frac{(X-1)\gamma_{X-1}}{k} \\
\end{bmatrix}}.
\end{align}
\end{subequations}
The matrix $\E$ is a bidiagonal matrix with its diagonal elements $1+\frac{(i-1)\gamma_{i-1}}{k}=\frac{k+(i-1)\gamma_{i-1}}{k}$ for $i=1,2,\cdots,X$. The super diagonal elements are given by $-\frac{j\gamma_j}{(1-\mu)k}$ for $j=1,2,\cdots,X-1$. Using the result for inverse of a bidiagonal matrix derived in \cite{Cha74}, we can write the $(i,j)$ element of $\E^{-1}$ as follows
\begin{equation}
e'_{i,j}= \begin{cases} 0 &\mbox{if } i > j, \\ 
\frac{k}{k+(i-1)\gamma_{i-1}}, & \mbox{if } i=j,\\
\frac{k}{k+(i-1)\gamma_{i-1}} \prod_{l=i}^{j-1}\left(\frac{1}{k+l\gamma_l} \frac{l\gamma_l}{(1-\mu)}\right), & \mbox{if } i<j.
 \end{cases}
\end{equation}
In matrix form, $\E^{-1}$ can be written as
\begin{align}
\E^{-1}={\begin{bmatrix}
1            &  \frac{\gamma_1}{k+\gamma_1}\frac{1}{1-\mu}                    & \cdots        & \prod_{l=1}^{j-1}\left(\frac{l\gamma_l}{k+l\gamma_l} \frac{1}{1-\mu}\right)           \\
0           & \frac{k}{k+\gamma_1}                      & \cdots          & \frac{k}{k+\gamma_1} \prod_{l=2}^{j-1}\left(\frac{l\gamma_l}{k+l\gamma_l} \frac{1}{1-\mu}\right) \\
\vdots                          & \vdots                                            & \vdots         & \vdots                                               \\
0           & 0                 & \cdots             & \frac{k}{k+(X-1)\gamma_{X-1}}   \\
\end{bmatrix}}.
\label{eqn:Einv}
\end{align}
Thus, we have determined the matrices $\E^{-1}$ and $\A_0^{-1}$ which can be used to compute $\A^{-1}=\E^{-1}\A_0^{-1}$.

\subsection{Expression of $\U \A^{-1}$}
\label{app:UAinv}
In \cite{GhS15}, where both production and degradation of protein do not have feedback regulation, we show that $\U\A^{-1}$ simplifies to a row vector with all of its elements being $1$. It turns out that even in the case when the protein degradation has a feedback, the same result is true. The calculations are as follows. Consider two matrices $\G$ and $\H$ such that $\A_e = \G \H$ where $\G$ is a $X \times X-1$ matrix

\footnotesize{
\begin{align}
\G= 
{\begin{bmatrix}
-\gamma_1           & 0          &  0            & \cdots        & 0                                            & 0              \\
\gamma_1            &  -2\gamma_2        &  0            & \cdots        & 0                                             & 0              \\
0           & 2\gamma_2           &  -3\gamma_3           & \cdots         & 0                                             & 0               \\
0           & 0            & 3\gamma_3               & \cdots       & 0                                            & 0 \\
\vdots                          & \vdots                                            & \vdots         & \vdots                                      & \vdots          \\
0           & 0           & 0       & \cdots         &  (X-2)\gamma_{X-2}      & -(X-1)\gamma_{X-1}     \\
0           & 0           & 0        & \cdots         & 0               & (X-1)\gamma_{X-1}   \\
\end{bmatrix}},
\end{align}}
\normalsize
while $\H$ is a $X-1 \times X$ matrix
\begin{align}
\H= 
 {\begin{bmatrix}
0           & -1          &  0            & \cdots        & 0                                            & 0              \\
0            &  0        &  -1           & \cdots        & 0                                             & 0              \\
\vdots                          & \vdots                                            & \vdots         & \vdots                                      & \vdots          \\
0           & 0           & 0       & \cdots         &  -1      & 0 \\
0           & 0           & 0        & \cdots         & 0               & -1  \\
\end{bmatrix}}.
\end{align}
We evoke the matrix inversion lemma to write $\A^{-1}$ as
\begin{subequations}
\begin{align}
\A^{-1} &= \left({\A_0} + \G  \H\right)^{-1} \\
                         &= {\A_0}^{-1}-\A_0^{-1}  \G \left(\I + \H\A_0^{-1}\G \right)^{-1}\H{\A_0}^{-1}.
\end{align}
This implies
\begin{align}
\U \A^{-1}= \U  \A_0^{-1} -  \U \A_0^{-1} \G \left(\I + \H\A_0^{-1}\G \right)^{-1}\H{\A_0}^{-1}.
\end{align}
\end{subequations}
Let us compute $\U \A_0^{-1}\G$:
\begin{subequations}
\begin{align}
& -{\begin{bmatrix}
1           \\  1                \\ 1 \\ 1    \\ \vdots       \\ 1             \\ 1
\end{bmatrix}}^T{\begin{bmatrix}
-\gamma_1           & 0                     & \cdots         & 0              \\
\gamma_1            &  -2 \gamma_2                    & \cdots        & 0              \\
0           & 2                     & \cdots          & 0               \\
0           & 0                       & \cdots        & 0 \\
\vdots                          & \vdots                                            & \vdots         & \vdots                                               \\
0           & 0                & \cdots              & -(X-1) \gamma_{X-1}  \\
0           & 0                 & \cdots             & X-1\gamma_{X-1}  \\
\end{bmatrix}}, \\
&= -{\begin{bmatrix}
0           &  0                   & \cdots        & 0            \\
\end{bmatrix}}.
\end{align}
\end{subequations}
Therefore, we can conclude that $\U \A^{-1}$ is in fact equal to $\U \A_0^{-1}$ which can be calculated by multiplying $\U$ and $\A_0^{-1}$ as shown below. 
\begin{subequations}
\begin{align}
&\U\A_0^{-1} = -\frac{\mu}{1-\mu} \times \nonumber \\ 
& {\begin{bmatrix}k \left(1-\mu\right)^{X} \\ k \left(1-\mu\right)^{X-1} \\ k \left(1-\mu \right)^{X-2}\\ \vdots \\  k \left(1-\mu\right)^2 \\ k \left(1-\mu\right)\end{bmatrix}}^T
{\begin{bmatrix}
\frac{1}{\mu k}                 & 0                                          & \cdots        & 0                                            & 0              \\
\frac{1}{k}            & \frac{1}{\mu k}                       & \cdots        & 0                                             & 0              \\
\frac{1}{k}           & \frac{1}{k}                              & \ddots         & 0                                             & 0               \\
\vdots                          & \vdots                                            & \vdots         & \vdots                                      & \vdots          \\
\frac{1}{k}   & \frac{1}{k}                     & \cdots         &  \frac{1}{\mu k} & 0    \\
\frac{1}{k}  & \frac{1}{k}                     & \cdots         & \frac{1}{k}                & \frac{1}{\mu k}  \\
\end{bmatrix}}\\
&= -\frac{\mu}{1-\mu}{\begin{bmatrix}\frac{\left(1-\mu \right)^X}{\mu} + \sum_{l=1}^{X-1}\left(1-\mu \right)^l \\ \frac{\left(1-\mu \right)^{X-1}}{\mu} + \sum_{l=1}^{X-2}\left(1-\mu \right)^l \\ \frac{\left(1-\mu \right)^{X-2}}{\mu} + \sum_{l=1}^{X-3}\left(1-\mu \right)^l \\ \vdots \\  \frac{\left(1-\mu\right)^2}{\mu} + \left(1-\mu \right) \\ \frac{1-\mu}{\mu} \end{bmatrix}}^T\\
&= -\frac{\mu}{1-\mu}{\begin{bmatrix}\frac{1-\mu}{\mu} & \frac{1-\mu}{\mu} &\frac{1-\mu}{\mu} & \cdots & \frac{1-\mu}{\mu}  & \frac{1-\mu}{\mu} \end{bmatrix}}\\
&=-{\begin{bmatrix} 1 & 1 & 1 & \cdots & 1  & 1 \end{bmatrix}}.
\label{eqn:UAinv}
\end{align}
\end{subequations}
\subsection{First two moments of $FPT$ for a birth-death process}
\label{sec:BDPformulas}
The matrix-based approach considered in this paper can be used to compute the formulas for moments of the first-passage time. These are given by:
\begin{subequations}
\begin{align}
&\left<FPT\right>=\sum_{i=1}^{X}\frac{1}{k}\left(1+\sum_{j=i+1}^{X}\prod_{l=i}^{j-1}\frac{l\gamma_l}{k}\right), \\
&\left<FPT^2\right>=2\sum_{i=1}^{X}\frac{1}{k}\left(\sum_{r=1}^{i}\tau_r+\sum_{j=i+1}^{X}\sum_{r=1}^{j}\tau_r \prod_{l=i}^{j-1}\frac{l\gamma_l}{k}\right),\;\;\text{where} \\
&\tau_i=\frac{1}{k}\left(1+\sum_{j=i+1}^{X}\prod_{l=i}^{j-1}\frac{l\gamma_l}{k}\right).
\end{align}
\end{subequations}
Though we have skipped the detailed steps here, these results are equivalent to the results derived in \cite{GFS15} and \cite{JoD08}. More specifically, if $\gamma_i=\gamma$, we get the results of \cite{GFS15}. The results in \cite{JoD08} are more general and simplify to our results if the birth rate is made constant while the death rate is taken as $\gamma_i$ instead of $i\gamma_i$. 
\bibliography{ACC16Ref}
\bibliographystyle{IEEEtran}
\end{document}